%% file: main.tex
\documentclass[letterpaper,11pt]{article}

\usepackage[utf8]{inputenc}
\usepackage[english]{babel}
\usepackage[margin=1in]{geometry}
\usepackage{amsmath,amsfonts,amssymb,amsthm,mathtools}
\usepackage{microtype}
\usepackage{needspace}
\usepackage[dvipsnames]{xcolor}
\definecolor{linkblue}{HTML}{0072BC}
\definecolor{citegreen}{HTML}{3C7F32}
\usepackage[colorlinks=true,pdfpagemode=UseNone,urlcolor=linkblue,
  linkcolor=linkblue,citecolor=citegreen,pdfstartview=FitH,
  pdftitle={Exponential Sampling Lower Bounds for Polynomial Sources},
  pdfauthor={Yan Zhong}]{hyperref}
\usepackage[nameinlink,noabbrev]{cleveref}
\newtheorem{theorem}{Theorem}
\newtheorem{lemma}[theorem]{Lemma}
\newtheorem{corollary}[theorem]{Corollary}

\newcommand{\F}{\mathbb F}
\newcommand{\U}{\mathbf U}
\newcommand{\E}{\mathbb E}
\newcommand{\newlang}[2]{%
  \newcommand{#1}{\ifmmode\mathop{\text{#2}}\nolimits\else#2\fi}}
\newlang{\Ov}{\textnormal{\textsc{Ov}}}
\newlang{\Bc}{\textnormal{\textsc{Bc}}}
\newlang{\Tc}{\textnormal{\textsc{Tc}}}
\newlang{\Tv}{\textnormal{\textsc{Tv}}}
\newlang{\Kl}{\textnormal{\textsc{Kl}}}
\newlang{\Ber}{\textnormal{\textsc{Ber}}}
\newlang{\PolySrc}{\textnormal{\textsc{PolySrc}}}

\title{Exponential Sampling Lower Bounds for Polynomial Sources}
\author{Yan Zhong\thanks{Johns Hopkins University,
  \texttt{\href{mailto:yanzhong.cs@gmail.com}{yanzhong.cs@gmail.com}}.}}
\date{September 5, 2026}

\begin{document}
\hypersetup{pageanchor=false}
\begin{titlepage}
\maketitle
\thispagestyle{empty}
\begin{abstract}
A degree-$d$ polynomial source is the joint distribution of polynomials
of degree at most $d$ over $\F_2$, evaluated on a common sequence of
independent uniform input bits of arbitrary length. Khodabandeh and
Shinkar (FOCS~'26) proved that the product distribution
$\Ber(1/3)^{\otimes N}$ has statistical distance $1-o(1)$
from every constant-degree polynomial source, and conjectured an
exponentially small overlap. Independently of Khodabandeh and Shinkar's
work, Byramji, Kane, Morris, and
Ostuni (RANDOM~'26) asked for an explicit target distribution at distance
$1-\exp(-N^{\Omega_d(1)})$.

We resolve both questions. For every fixed $d\geq1$, every degree-$d$
polynomial source has overlap at most $\exp(-c_dN)$ with
$\Ber(1/3)^{\otimes N}$, where $c_d>0$ is independent of
the input length. For quadratics, we obtain the explicit constant
$c_2=2^{-26}$. We amplify a result of Khodabandeh and Shinkar: for each
fixed degree, the acceptance probabilities of Boolean polynomials are
uniformly bounded away from any fixed non-dyadic $p\in(0,1)$.\footnote{A real
number is \emph{dyadic} if it can be written as $a/2^b$ for integers $a$
and $b\geq0$, and \emph{non-dyadic} otherwise.} This yields
exponentially small overlap with the corresponding product distributions
and extends to coordinates that are arbitrary Boolean functions of a
bounded number of bounded-degree polynomials.

We also present a quantitative refinement that gives a uniform deterministic
hierarchy between adjacent polynomial degrees with exponentially small overlap.
Appending the outputs of disjoint AND gates on $d+1$ inputs to the
uniform seed bits yields flat degree-$(d+1)$ target distributions of
entropy $k$, for
$\min\{k,N-k\}\geq2(d+1)$, with overlap
$\exp(-\Omega_d(\min\{k,N-k\}))$ against every degree-$d$ source.
For each fixed $d\geq1$, the dependence on $\min\{k,N-k\}$ is optimal up to
constant factors in the exponent among flat target distributions of entropy $k$.
The construction has locality $d+1$ and takes $O(N)$ field operations to
sample. For every fixed $0<\varepsilon<1$, this degree-dependent family handles degrees up to
$\frac{1-\varepsilon}{3}\log_2N$
with overlap $\exp(-N^{\varepsilon-o(1)})$ at entropy $k=\lfloor N/2\rfloor$.

The proof combines monotonicity of Gowers uniformity norms,
pairwise independence of points in a random affine cube, and bounds on
relative entropy.
\end{abstract}
\end{titlepage}
\pdfbookmark[1]{Contents}{contents}
\tableofcontents
\thispagestyle{empty}
\clearpage
\setcounter{page}{1}
\hypersetup{pageanchor=true}

\section{Introduction}

Sampling lower bounds compare a specified distribution $M$, called the
\emph{target distribution}, with outputs $Q(\U_s)$ of computationally
restricted samplers. The sampler may use any parametrization and any
number of random inputs. For example, for uniform $X\in\F_2^k$, the
joint distribution $(X,\operatorname{parity}(X))$, where
$\operatorname{parity}(x)=\bigoplus_{i=1}^k x_i$, can be sampled with
each output bit depending on at most two random input bits, even though
parity is hard to compute in $\mathsf{AC}^0$~\cite[Section~1]{LovettViola11}.
Such examples motivate the study of the complexity of distributions~\cite{Viola12}.
We consider polynomial samplers over $\F_2$.

A degree-$d$ polynomial source on $N$ bits is a distribution $Q(\U_s)$,
where every coordinate of $Q:\F_2^s\to\F_2^N$ has total Boolean degree
at most $d$. The seed length $s$ is unrestricted. We denote this class by
$\PolySrc_d(N)$ and write
\[
 \Ov(P,M)=1-\|P-M\|_{\Tv}=\sum_y\min\{P(y),M(y)\}
\]
for the overlap of two distributions. Exponentially small overlap means
statistical distance exponentially close to one. The distance is
statistical: the complexity restriction applies to the sampler, while the
distinguishing event is unrestricted.

\paragraph{From locality to polynomial degree.}
A sampler is $\ell$-local if each output bit depends on at most $\ell$
input bits. Even independent biased bits can be hard to sample in this
model. Kane, Ostuni, and Wu~\cite[Theorem~1.10]{KOW24} proved that,
for every fixed $\ell$, any $\ell$-local sampler has overlap at most
$\exp(-N/2^{O(\ell^2)})$ with $\Ber(1/3)^{\otimes N}$,
with no restriction on its input length. Every $\ell$-local Boolean
function has degree at most $\ell$, but a low-degree polynomial may
depend on every input bit. Thus the same target distribution raises a stronger question:
does an exponential lower bound follow from a bound on degree alone?

The locality proof uses two properties that degree alone does not give.
Each output mean is a multiple of $2^{-\ell}$ and hence stays away from
$1/3$. The bounded input neighborhoods allow a restriction argument to
produce many independent outputs. A degree bound gives neither small
neighborhoods nor denominators bounded in terms of the degree.\footnote{For
independent uniform bits $x_1,y_1,\ldots,x_r,y_r$, the quadratic
$\bigoplus_{j=1}^r x_jy_j$ has acceptance probability
$(1-2^{-r})/2$. Its denominator grows with $r$, although its degree is two.}
A proof from a degree bound alone must recover a uniform gap from $1/3$
and amplify it without bounded input neighborhoods.

\paragraph{Extractors and the entropy barrier.}
De and Watson~\cite{DeWatson12} combine extraction with a structural
decomposition of local samplers to obtain constant-distance lower bounds.
Viola~\cite{Viola14} develops a general reduction from extractors for
circuit sources. For a one-bit extractor $E$, the distribution
$(\U_n,E(\U_n))$ is uniform on the graph of $E$ and always satisfies
$b=E(x)$. If a candidate sampler's last output bit is fixed while its
first $n$ output bits remain a source handled by $E$, extraction forces
roughly half its outputs off the graph. However, closeness to a flat
target distribution does not guarantee high min-entropy. Adding a small
point mass already creates a heavy atom. Thus an extractor's min-entropy
requirement can fail even for a candidate sampler with substantial
overlap with the target distribution.
Removing the heavy outcomes by conditioning need not preserve the class
of samplable sources, an obstruction already identified by
Chattopadhyay, Goodman, and Zuckerman~\cite[full version, Section~2.3]{CGZ22}
in their work on small-space sampling.

Byramji, Kane, Morris, and Ostuni~\cite{BKMO26} address this issue with
robust extractors. For an entropy threshold $k$, an outcome is
\emph{light} if its source probability is at most $2^{-k}$.
In addition to ordinary extraction, one fixed output of an $m$-bit
extractor must receive probability mass at most $2^{-m}+\delta$
from the light outcomes of every source in the class, where
$\delta$ is the robustness error~\cite[Definition~3.1]{BKMO26}.
This bounds original probability mass, without renormalizing after
conditioning. Their
reduction turns this guarantee into sampling lower bounds. By adapting
the polynomial-source extractors of Chattopadhyay, Goodman, and
Gurumukhani~\cite{CGG24}, they obtain explicit distributions at distance
$1-o(1)$ from every constant-degree polynomial source.

\paragraph{The prescribed product distribution.}
Khodabandeh and Shinkar~\cite{KS26} obtained distance $1-o(1)$ for
$\Ber(1/3)^{\otimes N}$ itself. Their work was independent of and
concurrent with that of Byramji, Kane, Morris, and Ostuni~\cite{BKMO26}.
They prove that every degree-at-most-$d$ Boolean polynomial $q$ satisfies
$|\Pr[q(\U_s)=1]-1/3|\geq\delta_d>0$, independently of
$s$~\cite[Theorem~6.4 and Remark~6.5]{KS26}. They then use a sunflower
argument to regularize pairs of output polynomials and apply Chebyshev's
inequality to obtain the sampling lower bound.

Khodabandeh and Shinkar conjectured overlap $\exp(-\Omega_d(N))$ for
this target distribution, including coordinates of bounded
$\operatorname{rank}_d$~\cite[Sections~1.2 and~2.6]{KS26}.
Here $\operatorname{rank}_d(q)\leq r$ means that $q$ is a Boolean function
of at most $r$ degree-at-most-$d$ polynomials.
Byramji, Kane, Morris, and Ostuni asked for an explicit target distribution with overlap
$\exp(-N^{\Omega_d(1)})$~\cite[Section~5, Question~2]{BKMO26}.

\subsection{Our Results}

Our main result proves the conjectured exponential bound for the
product distribution $\Ber(1/3)^{\otimes N}$:
for every fixed $d\geq1$, there is $c_d>0$ such that
\begin{equation}
 (2/3)^N\leq\sup_{X\in\PolySrc_d(N)}
 \Ov\bigl(X,\Ber(1/3)^{\otimes N}\bigr)
 \leq e^{-c_dN}.
 \label{eq:intro-one-third}
\end{equation}
There is no restriction on the number of random inputs or on the joint
structure of the coordinate polynomials.
The lower bound is witnessed by the constant sampler $Q\equiv0^N$,
so the exponential order in $N$ is optimal for each fixed $d$.
\Cref{thm:unrestricted-fixed-degree-one-third-cube} proves this statement
with $c_d=\delta_d^6$, where $\delta_d$ is the constant in
\Cref{thm:ks-acceptance-probabilities} for degree-at-most-$d$ polynomials
and $p=1/3$. For quadratics, their possible acceptance probabilities
give $c_2=2^{-26}$
(\Cref{cor:unrestricted-quadratic-one-third-cube}). The same method proves
exponential separation for every fixed non-dyadic parameter $p$, and for
coordinates that are arbitrary Boolean functions of a bounded number of
bounded-degree polynomials. For any input distribution $\nu$, it also
lower-bounds the sum of its relative entropy from the uniform seed
distribution and the output's relative entropy from
$\Ber(1/3)^{\otimes N}$
(\Cref{cor:cube-hellinger-entropy-separation}).

The proof is a general amplification principle. Suppose a class of Boolean
functions is closed under affine restriction and there is a constant
$\delta>0$ such that every function in the class has mean at distance at
least $\delta$ from $p$. Then a map whose coordinates belong to the class
has overlap $\exp(-\Omega_{p,\delta}(N))$ with
$\Ber(p)^{\otimes N}$. More precisely, it suffices to have
the same lower bound on the distance from $p$ for each coordinate mean
on every affine cube of one suitably chosen dimension
(\Cref{thm:affine-cube-amplification}). Theorem~6.4 and Remark~6.5
of~\cite{KS26} verify this condition for degree-at-most-$d$ polynomials
when $p=1/3$. Our contribution is the resulting exponential separation
for an arbitrary polynomial sampler.

To construct the adjacent-degree hierarchy, we consider
$\Ber(p)^{\otimes n}$ with $p=2^{-(d+1)}$, which can be sampled
by AND gates on disjoint blocks of $d+1$ uniform input bits.
For each fixed $d$, Khodabandeh and Shinkar observed that this
distribution has statistical distance $1-o(1)$ from every
degree-at-most-$d$ polynomial source~\cite[Section~1.2]{KS26}.
Every degree-at-most-$d$ Boolean polynomial has acceptance probability
either zero or at least $2^{-d}=2p$, so its acceptance probability differs
from $p$ by at least $p$. A variance-sensitive form of
the amplification theorem gives the uniform estimate
\begin{equation}
 \sup_{X\in\PolySrc_d(n)}
 \Ov\bigl(X,\Ber(2^{-(d+1)})^{\otimes n}\bigr)
 \leq\exp(-2^{-3d-10}n).
 \label{eq:intro-and-product}
\end{equation}
Including the seed bits in the output makes the sampler injective and
its output distribution flat. Projection onto the AND outputs preserves
polynomial degree, so Eqn.~\eqref{eq:intro-and-product} also bounds the
overlap with this graph distribution.
For every entropy $k$ with $\min\{k,N-k\}\geq2(d+1)$, we obtain a graph
distribution $D=(\U_k,F(\U_k))$ of degree exactly $d+1$ satisfying
\begin{equation}
 \sup_{X\in\PolySrc_d(N)}\Ov(X,D)
 \leq\exp\left(-\frac{2^{-3d-11}}{d+1}\min\{k,N-k\}\right).
 \label{eq:intro-entropy-hierarchy}
\end{equation}
Every output depends on at most $d+1$ inputs. The construction is uniform,
deterministic, and uses $O(N)$ field operations to sample, even when $d$
grows. Conversely, every flat target distribution of entropy $k$ has overlap at least
$2^{-\min\{k,N-k\}}$ with some degree-$d$ source. Thus the entropy
dependence in Eqn.~\eqref{eq:intro-entropy-hierarchy} is optimal up to a factor
depending only on $d$ (\Cref{cor:cube-and-block-hierarchy}). For this
degree-dependent family at entropy $k=\lfloor N/2\rfloor$, the explicit
constants allow
\[
 d\leq\frac{1-\varepsilon}{3}\log_2N,
 \qquad \Ov(X,D)\leq\exp(-N^{\varepsilon-o(1)}),
\]
for every fixed $0<\varepsilon<1$
(\Cref{cor:cube-growing-degree}).

\subsection{Relation to Previous Sampling Results}

\paragraph{Local sampling and amplification.}
The locality program extends well beyond biased coins. Filmus, Leigh,
Riazanov, and Sokolov~\cite{FLRS23} proved lower bounds for sampling
sublinear-weight Hamming slices using decision forests.
Kane, Ostuni, and Wu proved strong locality lower bounds for
constant-density slices~\cite{KOW24} and classified the uniform
symmetric distributions that constant-locality samplers can approximate:
as the error tends to zero, the target distribution must approach the uniform cube,
one of the two uniform parity classes, one of the two constant strings,
or the fair mixture of the two constants~\cite{KOW25}.
For general symmetric target distributions, accurate local sampling forces
approximation by mixtures of dyadic product distributions and uniform
parity classes~\cite{KOW26}.
Grier, Kane, Morris, Ostuni, and Wu~\cite[full version, Theorem~2.11]{GKMOW26}
also prove a direct-product theorem: at fixed locality, a constant lower
bound with arbitrary binary product seeds amplifies exponentially for
powers of a target distribution of fixed block length. Their proof isolates independent blocks
by fixing shared input bits using a graph elimination argument. The allowance of biased seeds
requires different target distributions: the one-third product can be sampled exactly
from one-third-biased seeds. In our setting the dependency graph can
be dense. We use the fact that affine restriction preserves
polynomial degree. Random affine cubes supply pairwise independent
evaluations without requiring disjoint input neighborhoods.

\paragraph{The sampler model and seed length.}
Lovett and Viola~\cite{LovettViola11} and Beck, Impagliazzo, and
Lovett~\cite{BeckImpagliazzoLovett12} obtain strong sampling lower bounds
against $\mathsf{AC}^0$ using linear codes. Those target distributions are themselves
degree-one polynomial sources. Viola~\cite{Viola16} similarly separates
some quadratic graph distributions from polynomial-size $\mathsf{AC}^0$
samplers. Another distinction is the number of random inputs:
Horacsek, Lee, Shinkar, Viola, and Zhou~\cite{HLSVZ26} study tradeoffs
between locality, seed length, and error for product distributions.
Their lower bounds for $\Ber(1/4)^{\otimes N}$ require a short seed.
With $2N$ random bits, independent AND gates sample it exactly.
Our degree lower bounds allow arbitrarily many random inputs.

\paragraph{Extraction and degree hierarchies.}
The degree of an extractor and the degree of the source it handles are
different parameters. Low-degree extraction from local sources was
studied in~\cite{ACGLR22}. Random low-degree polynomials also extract
from polynomial sources~\cite{GGHNY24,AGMR25}.
Golovnev et al.~\cite[Section~1.5]{GGHNY24} deduced a degree-$O(d)$
polynomial $g$ whose graph distribution $(\U_n,g(\U_n))$ cannot be
sampled exactly by a degree-$d$ map. The dyadic-product observation
of~\cite[Section~1.2]{KS26} already separates adjacent degrees.
We obtain exponentially small overlap, explicit degree dependence,
and uniform deterministic flat distributions at prescribed entropy.
Every one-bit graph distribution
$(\U_k,g(\U_k))$ has overlap $1/2$ with $\U_{k+1}$.
Such a target distribution cannot give exponentially small overlap against polynomial
samplers, regardless of how hard $g$ is to compute.

\paragraph{Affine extractors and sampling lower bounds.}
Affine sources are precisely the degree-at-most-one polynomial
sources. Directional affine extractors yield average-case lower
bounds for strongly read-once linear branching
programs~\cite{GPT22,LZ24Directional}. Li and
Zhong~\cite{LZ24NonMalleable} also obtain such lower bounds using
affine non-malleable extractors.
These results concern the complexity of computing explicit
functions. In our setting, the question is whether an explicit
target distribution can be sampled by a low-degree polynomial map.
Even for a target distribution of the form $(\U_k,F(\U_k))$,
a competing sampler may generate both components jointly as
$(A(\U_s),B(\U_s))$, without evaluating $F$ on a supplied input.
Consequently, computational lower bounds for $F$ do not directly
give the required sampling lower bound: one must control the
overlap with every degree-$d$ map $(A,B)$ and every seed length~$s$.
We establish such bounds for our target distributions through
weighted affine-cube amplification.

\subsection{Proof Overview}

Coordinate means alone do not imply small overlap. For
$M=\Ber(p)^{\otimes N}$ and
$P=\tfrac12 M+\tfrac12\delta_{0^N}$, where $\delta_{0^N}$ is the
point mass at $0^N$, every coordinate of $P$ has mean $p/2$, but
$\Ov(P,M)\geq1/2$. Now fix $p$ and $\delta>0$ such that every
degree-at-most-$d$ Boolean polynomial has acceptance probability at
distance at least $\delta$ from $p$. We sample an affine cube in the
input space and test the empirical mean of each output coordinate on it.
Every coordinate restriction remains a polynomial of degree at most $d$.
Hence its empirical mean differs from $p$ by at least $\delta$, on
\emph{every} cube, including degenerate cubes.

In his work on sampling in the cell-probe model,
Viola~\cite[Lemma~15 and Theorem~14]{Viola23} converts sufficiently large
overlap with a flat target distribution into a conditioning event with
controlled min-entropy,
then restricts the seed set to obtain approximate pairwise independence
among output coordinates.
We use a change of measure defined by the common mass of the source
and the nonflat product target distribution.
Let $P=Q(\U_s)$, $M=\Ber(p)^{\otimes N}$, and let their overlap
be $\omega$. Define weights on the seed space by
\[
 f(z)=\frac{\min\{P(Q(z)),M(Q(z))\}}{P(Q(z))},
 \qquad 0\leq f\leq1,\qquad \E f=\omega.
\]
For a random affine $t$-cube $(X_v)_{v\in\F_2^t}$, multiply the
probability of each choice of cube parameters by $\prod_v f(X_v)$ and
normalize. Denote the resulting distribution by $\Pi$.
The normalizing constant is at least $\omega^K$, where
$K=2^t$. For nonnegative functions, this follows from monotonicity of
Gowers uniformity norms~\cite[Definition~2.2 and Lemma~2.4]{ViolaWigderson08}. We include an elementary induction proof.
The new step is to use the overlap weights to control the output
marginals under this change of measure.

For distinct labels $v,w$, the points $X_v,X_w$ of the original random
cube are independent uniform inputs.
Dropping the other factors, which are at most one, shows that each pair
of output vectors $Q(X_v),Q(X_w)$ with $v\ne w$ under the new
distribution has density at most $\omega^{-K}$
relative to $M\otimes M$. Its relative entropy is therefore at most
$K\ln(1/\omega)$. Decomposing this entropy across output coordinates
and applying Pinsker's inequality bounds the squared empirical means:
\[
 N\delta^2
 \leq \E_\Pi\sum_{i=1}^N
       \left(\frac1K\sum_v Q_i(X_v)-p\right)^2
 \leq\frac{Np(1-p)}K
       +\max\{p,1-p\}\sqrt{\frac{NK\ln(1/\omega)}2}.
\]
Choose a fixed $K>p(1-p)/\delta^2$. This inequality forces
$\ln(1/\omega)=\Omega_{p,\delta}(N)$. To obtain useful constants when $p$ is small, we replace
Pinsker's inequality by a bound on expectations in terms of relative
entropy and variance. The resulting estimate
Eqn.~\eqref{eq:intro-and-product} then yields the adjacent-degree hierarchy
by coordinate projection.

\Needspace{9\baselineskip}
\subsection{Organization of the Paper}
\Cref{sec:preliminaries} gives the notation, probabilistic tools, and
the results of Khodabandeh and Shinkar on acceptance probabilities.
\Cref{sec:ks-cube-amplification} proves affine-cube amplification and
the product-distribution lower bounds, including the bounded-rank
extension, the explicit quadratic bound, and bounds on Hellinger
affinity and relative entropy.
\Cref{sec:cube-local-hierarchy} bounds the overlap with products of AND
outputs and constructs the adjacent-degree hierarchy at prescribed
entropy, including the guarantee for growing degrees.

\input{prelim}
\input{ks_cube_amplification}
\input{ks_cube_entropic_corollary}
\input{cube_local_hierarchy}

\section*{Acknowledgments}
This work was conducted while the author was visiting the Simons Institute for the Theory of Computing on the program Pseudorandomness and High-Dimensional Expansion as a visiting student researcher. The author thanks the institute for its hospitality
and support.

\paragraph{AI Disclosure.}
The author supplied the initial manuscript and research ideas, formulated
the questions and proposed directions for strengthening the results, guided
the proof development, and directly revised the manuscript. OpenAI's
ChatGPT/Codex assisted with developing and checking proofs, literature
checks, drafting, revisions following the author's suggestions, and the
Lean~4 formalization.
A companion Lean~4 formalization verifies the main results relative to
explicitly stated inputs from the literature. Its scope and dependencies
are documented with the code.\footnote{Lean formalization:
\url{https://github.com/yanzhong-bw/SamplingLowerBound}.}
The author retains responsibility for the mathematical claims and the
final manuscript.

\bibliographystyle{alpha}
\bibliography{references}
\end{document}

%% file: prelim.tex
\section{Preliminaries}
\label{sec:preliminaries}

\subsection{Polynomial Sources and Graph Distributions}
\label{sec:prelim-sources}
Write $[n]=\{1,\ldots,n\}$ and let $\U_s$ be the uniform law on
$\F_2^s$. Every Boolean function has a unique multilinear polynomial
representation over $\F_2$. Degree refers to this representation.
A map has degree at most $d$ if each coordinate does, and
$\PolySrc_d(n)$ consists of the laws $Q(\U_s)$ of such maps, with
arbitrary $s\geq0$. Substitution of affine forms cannot increase degree,
nor can projection onto output coordinates. In real-valued expressions,
Boolean outputs are identified with $0$ and $1$.

\paragraph{Finite probability distributions.}
All distributions considered here have finite support. For a map $T$,
$T(R)$ denotes the law of $T(X)$ when $X\sim R$. The law
$\Ber(p)$ assigns probability $p$ to $1$. Tensor products
denote independent coordinates. We use
\[
 \|R-M\|_{\Tv}=\frac12\sum_y|R(y)-M(y)|,
 \qquad \Ov(R,M)=1-\|R-M\|_{\Tv}.
\]
Grouping terms over the fibers of $T$ gives
$\|T(R)-T(M)\|_{\Tv}\leq\|R-M\|_{\Tv}$, so overlap increases
under deterministic postprocessing. A distribution is \emph{flat of
entropy $k$} if it is uniform on $2^k$ points. This parameter is measured
in bits. All other logarithms and entropy quantities use natural
logarithms, including $H(R)=-\sum_yR(y)\ln R(y)$.

\paragraph{Graphs of functions and locality.}
For $F:\F_2^k\to\F_2^m$, its graph is the set
$\{(x,F(x)):x\in\F_2^k\}$. We call the uniform law
$(\U_k,F(\U_k))$ the \emph{graph distribution} of $F$.
It is flat of entropy $k$, since its first $k$ coordinates identify the
input. A sampler is $\ell$-local if each output coordinate depends on
at most $\ell$ input bits.

\subsection{Acceptance Probabilities}
\label{sec:prelim-acceptance}
For $q:\F_2^s\to\F_2$, the distance of $q(\U_s)$ from a Bernoulli
distribution is
\begin{equation}
 \left|\Pr[q(\U_s)=1]-p\right|
 =\|q(\U_s)-\Ber(p)\|_{\Tv}.
 \label{eq:acceptance-probability-distance}
\end{equation}
We use positive lower bounds on this quantity that hold for every
function in a specified class, independently of its seed length $s$.

Following~\cite[Definition~3.7]{KS26}, the \emph{degree-$d$ rank} of a
Boolean function $q$, denoted by $\operatorname{rank}_d(q)$, is the smallest
positive integer $r$ such that $q=\Gamma(q_1,\ldots,q_r)$ for some Boolean
function $\Gamma$ and degree-at-most-$d$ polynomials $q_1,\ldots,q_r$.
We include constant functions by allowing $\Gamma$ to be constant.
We use the following results of Khodabandeh and Shinkar on acceptance
probabilities.

\begin{theorem}[Khodabandeh--Shinkar]
\label{thm:ks-acceptance-probabilities}
For every pair of integers $d,r\geq1$ there is $\delta_{d,r}>0$, independent of $s$,
such that every $q:\F_2^s\to\F_2$ with
$\operatorname{rank}_d(q)\leq r$ satisfies~\cite[Theorem~6.4]{KS26}
\[
 \left|\Pr[q(\U_s)=1]-1/3\right|\geq\delta_{d,r}.
\]
In particular, degree-at-most-$d$ polynomials admit
$\delta_d=\delta_{d,1}$~\cite[Remark~6.5]{KS26}.
For every fixed non-dyadic $p\in(0,1)$ and $d\geq1$, there is likewise
$\delta_{d,p}>0$, independent of $s$, such that every degree-at-most-$d$
polynomial satisfies~\cite[Theorem~A.8]{KS26}
\[
 \left|\Pr[q(\U_s)=1]-p\right|\geq\delta_{d,p}.
\]
\end{theorem}
Khodabandeh and Shinkar derive these bounds using polynomial
regularization and the bias-versus-rank theorem of Kaufman and
Lovett~\cite{KaufmanLovett08}. Our amplification uses the resulting
acceptance-probability bounds as input.

\subsection{Relative Entropy and Pinsker's Inequality}
\label{sec:prelim-entropy}
Write $R\ll M$ when $M(y)=0$ implies $R(y)=0$, and define
\[
 D_{\Kl}(R\Vert M)=\sum_{y:R(y)>0}R(y)\ln\frac{R(y)}{M(y)}
\]
when $R\ll M$, and $+\infty$ otherwise. In particular, $0\ln0=0$.
If $R(y)\leq C M(y)$ for all $y$, then
$D_{\Kl}(R\Vert M)\leq\ln C$.
For a law $R$ on $\prod_{i=1}^n\Omega_i$, with marginals $R_i$, and
a product reference $M=\bigotimes_iM_i$, the identity
\begin{equation}
 D_{\Kl}(R\Vert M)
 =\Tc(R)+\sum_iD_{\Kl}(R_i\Vert M_i),
 \qquad
 \Tc(R):=D_{\Kl}(R\Vert\bigotimes_iR_i)
 \label{eq:prelim-entropy-tensorization}
\end{equation}
follows by expanding the logarithm. Nonnegativity of relative entropy
therefore bounds the sum of marginal divergences by the joint divergence.
Pinsker's inequality~\cite[Theorem~31]{vanErvenHarremoes14}, in this normalization, states that
$\|R-M\|_{\Tv}\leq\sqrt{D_{\Kl}(R\Vert M)/2}$.
Consequently, if $g_i:\Omega_i\to\mathbb R$ each have range of length
at most $a$, then
\begin{equation}
 \left|\sum_i\bigl(\E_{R_i}g_i-\E_{M_i}g_i\bigr)\right|
 \leq a\sqrt{\frac n2 D_{\Kl}(R\Vert M)}.
 \label{eq:prelim-coordinate-test-bound}
\end{equation}
Indeed, each summand in absolute value is at most
$a\|R_i-M_i\|_{\Tv}$. The bound follows from Pinsker's inequality,
Cauchy--Schwarz, and Eqn.~\eqref{eq:prelim-entropy-tensorization}.

\subsection{Affine Cubes and Affine Subspaces}
\label{sec:prelim-affine-cubes}
An affine $t$-cube is a parametrized affine map
\[
 \phi:\F_2^t\to\F_2^s,\qquad
 \phi(v)=x+\sum_{j=1}^t v_jh_j.
\]
Its image is the affine subspace
$A=x+\operatorname{span}\{h_1,\ldots,h_t\}$, of dimension
$r=\operatorname{rank}(h_1,\ldots,h_t)\leq t$.
It is $t$-dimensional exactly when the directions are independent.
We retain all $2^t$ labels $v\in\F_2^t$: each point of $A$ occurs
$2^{t-r}$ times. Thus the uniform average over labels equals the
uniform average on $A$, whereas products retain these multiplicities.
This is the affine parametrization used to define Gowers uniformity
norms~\cite[Definition~2.2]{ViolaWigderson08}.

A uniform random cube samples $x,h_1,\ldots,h_t$ independently and
uniformly, with no conditioning on rank. For distinct labels $v,w$,
$X_v=\phi(v)$ and $X_w=\phi(w)$ are independent and uniform:
choose $j$ with $v_j\ne w_j$ and condition on the other directions.
The map $(x,h_j)\mapsto(X_v,X_w)$ is then a bijection.
For every fixed $\phi$, the pullback $q\circ\phi$ has degree at most
$d$ whenever $q$ does. Bounds on $\operatorname{rank}_d$ are preserved as well.
This agrees with affine-subspace restriction~\cite[Fact~3.8]{KS26}
and also covers noninjective parametrizations. Hence
\Cref{thm:ks-acceptance-probabilities} applies to the empirical mean on every cube.

%% file: ks_cube_amplification.tex

\section{Affine-Cube Amplification for Product Distributions}
\label{sec:ks-cube-amplification}

\subsection{The Amplification Theorem}
\label{sec:cube-amplification-theorem}

We first record the inequality used to bound the normalizing constant
of the change of measure. For nonnegative real-valued $f$, the quantity $Z$ in
\Cref{lem:weighted-affine-cube-lower-bound} equals $U_t(f)$
as defined in~\cite[Definition~2.2]{ViolaWigderson08}.
The lower bound follows from their monotonicity
lemma~\cite[Lemma~2.4]{ViolaWigderson08}. We include an
elementary proof for completeness. We then combine it with a change of
measure defined by the common mass of the source and target
distributions. This turns a uniform lower bound on the distance of
coordinate means from $p$ on affine restrictions into exponentially small
overlap with $\Ber(p)^{\otimes n}$.

\begin{lemma}[Weighted affine-cube lower bound]
\label{lem:weighted-affine-cube-lower-bound}
Let $f:\F_2^s\to[0,1]$, let $t\ge1$, and put $K=2^t$.
For a uniform random affine $t$-cube, allowing degeneracies,
\[
 Z:=\E_{x,h_1,\ldots,h_t}
       \prod_{v\in\F_2^t}f\left(x+\sum_{j=1}^t v_jh_j\right)
 \ge (\E f)^K.
\]
\end{lemma}

\begin{proof}
Write $C_t(f)$ for the left side. For $t=1$, the points $x,x+h_1$ are
independent and uniform, so $C_1(f)=(\E f)^2$. For $t\ge2$, set
$f_h(z)=f(z)f(z+h)$. Grouping the labels $v\in\F_2^t$ by $v_t$, and
then using induction and Jensen's inequality, gives
\[
 C_t(f)=\E_hC_{t-1}(f_h)
 \ge\E_h(\E_zf(z)f(z+h))^{2^{t-1}}
 \ge(\E_{h,z}f(z)f(z+h))^{2^{t-1}}
 =(\E f)^{2^t}.
\]
\end{proof}

\begin{theorem}[Exponential separation from a product distribution]
\label{thm:affine-cube-amplification}
Fix $p\in(0,1)$, an integer $t\ge1$, and $\delta>0$.
Set $K=2^t$ and $m=\max\{p,1-p\}$.  Let
$Q:\F_2^s\to\F_2^n$ be any map with the following property:
for every $x,h_1,\ldots,h_t\in\F_2^s$ and every $i\in[n]$,
\begin{equation}
 \left|\frac1K\sum_{v\in\F_2^t}
 Q_i\left(x+\sum_{j=1}^t v_jh_j\right)-p\right|\ge\delta.
 \label{eq:cube-coordinate-mean-separation}
\end{equation}
If $\delta^2>p(1-p)/K$, then, for arbitrary seed dimension $s$,
\begin{equation}
 \Ov\left(Q(\U_s),
                \Ber(p)^{\otimes n}\right)
 \le\exp\left(-\frac{2}{m^2K}
       \left(\delta^2-\frac{p(1-p)}K\right)^2n\right).
 \label{eq:cube-exponential-overlap-bound}
\end{equation}
\end{theorem}

\begin{proof}
Put $P=Q(\U_s)$, $M=\Ber(p)^{\otimes n}$,
$\omega=\Ov(P,M)$, and $L=\ln(1/\omega)$.
The case $\omega=0$ is immediate.  Define
\[
 f(z)=\frac{\min\{P(Q(z)),M(Q(z))\}}{P(Q(z))}.
\]
Every denominator is positive, $0\le f\le1$, and $\E f=\omega$.
Write $X_v=x+\sum_jv_jh_j$ for a uniform random affine cube and let
$\Pi$ be the probability law on its parameters with density
\[
 \frac{d\Pi}{d\U_{s(t+1)}}
 =\frac1Z\prod_{v\in\F_2^t}f(X_v),
 \qquad Z=\E\prod_v f(X_v)\ge\omega^K.
\]
The last inequality is \Cref{lem:weighted-affine-cube-lower-bound}.
Under $\Pi$, put $Y_v=Q(X_v)$.

We first bound the relative entropy of the laws of $Y_v$ and
$(Y_v,Y_w)$ for distinct labels $v,w$. The pair $(X_v,X_w)$ is
independent and uniform under the original uniform distribution on cube
parameters.  Since all other factors $f$ are at most
one, for every $a,b\in\F_2^n$ we have
\begin{align*}
 \Pr_\Pi[Y_v=a,Y_w=b]
 &\le\frac1Z\E[
 \mathbf1_{Q(X_v)=a,Q(X_w)=b} f(X_v)f(X_w)]\\
 &=\frac1Z\min\{P(a),M(a)\}\min\{P(b),M(b)\}
 \le\frac{M(a)M(b)}Z.
\end{align*}
Consequently,
\begin{equation}
 D_{\Kl}(\mathcal L_\Pi(Y_v,Y_w)\Vert M\otimes M)
 \le\ln(1/Z)\le K\,L.
 \label{eq:cube-pair-kl-bound}
\end{equation}
The same argument with one retained factor gives
$D_{\Kl}(\mathcal L_\Pi(Y_v)\Vert M)\le K\,L$.

Apply Eqn.~\eqref{eq:prelim-coordinate-test-bound} using the relative-entropy bound in
Eqn.~\eqref{eq:cube-pair-kl-bound}, grouping the two output bits at each
coordinate. For $v\ne w$, this gives
\begin{equation}
 \left|\sum_{i=1}^n\E_\Pi[(Y_{v,i}-p)(Y_{w,i}-p)]\right|
 \le m\sqrt{nK\,L/2}.
 \label{eq:cube-offdiagonal-second-moment}
\end{equation}
Indeed, the function $(a,b)\mapsto(a-p)(b-p)$ has mean zero under two
independent $\Ber(p)$ bits and range of length exactly
$m$. Thus Eqn.~\eqref{eq:prelim-coordinate-test-bound} applies with
range parameter $a=m$.
Similarly, because $(a-p)^2$ has mean $p(1-p)$ under $\Ber(p)$ and range
length $|1-2p|\le m$, the bound
$D_{\Kl}(\mathcal L_\Pi(Y_v)\Vert M)\leq KL$ gives
\begin{equation}
 \sum_{i=1}^n\E_\Pi[(Y_{v,i}-p)^2]
 \le np(1-p)+m\sqrt{nK\,L/2}.
 \label{eq:cube-diagonal-second-moment}
\end{equation}

Now put $A_i=K^{-1}\sum_vY_{v,i}$.  Expanding the square and using
Eqn.~\eqref{eq:cube-offdiagonal-second-moment} and
Eqn.~\eqref{eq:cube-diagonal-second-moment} yields
\[
 \sum_{i=1}^n\E_\Pi[(A_i-p)^2]
 \le\frac{np(1-p)}K+m\sqrt{nK\,L/2}.
\]
On every cube, hypothesis Eqn.~\eqref{eq:cube-coordinate-mean-separation} makes the
left side at least $n\delta^2$.  Hence
\[
 L\ge\frac{2n}{m^2K}
       \left(\delta^2-\frac{p(1-p)}K\right)^2,
\]
which proves Eqn.~\eqref{eq:cube-exponential-overlap-bound}.
\end{proof}

\subsection{Exponential Separation from Product Distributions}
\label{sec:cube-product-applications}

\begin{theorem}[The one-third product distribution]
\label{thm:unrestricted-fixed-degree-one-third-cube}
For every fixed $d\ge1$, there is $c_d>0$ such that, for all seed
lengths $s$, output lengths $n$, and degree-at-most-$d$ maps
$Q:\F_2^s\to\F_2^n$,
\begin{equation}
 \Ov\left(Q(\U_s),
       \Ber(1/3)^{\otimes n}\right)
 \le\exp(-c_dn).
 \label{eq:unrestricted-degree-d-exp-one-third}
\end{equation}
One may take $c_d=\delta_d^6$, with $\delta_d$ as in
\Cref{thm:ks-acceptance-probabilities}.
Thus this proves the exponential-error assertion in the first open
problem of~\cite[Section~2.6]{KS26} for every fixed degree.
\end{theorem}

\begin{proof}
For every cube parametrization $\phi:\F_2^t\to\F_2^s$,
$Q_i\circ\phi$ is a degree-at-most-$d$ Boolean polynomial on $t$
variables. Applying \Cref{thm:ks-acceptance-probabilities} to this pullback gives
Eqn.~\eqref{eq:cube-coordinate-mean-separation} with $\delta=\delta_d$,
including when $\phi$ is noninjective. We may assume $\delta_d\le1/3$.
Take $K=2^{\lceil\log_2(\delta_d^{-2})\rceil}$, for which
$\delta_d^{-2}\le K<2\delta_d^{-2}$.  Then
\[
 \delta_d^2-\frac{2}{9K}\ge\frac79\delta_d^2,
 \qquad
 \frac{2}{(2/3)^2K}
  \left(\delta_d^2-\frac{2}{9K}\right)^2
 \ge\frac{49}{36}\delta_d^6\ge\delta_d^6.
\]
Applying \Cref{thm:affine-cube-amplification} with $p=1/3$,
$\delta=\delta_d$, and this choice of $K$ gives the claimed
bound with $c_d=\delta_d^6$.
\end{proof}

\paragraph{Bounded degree-$d$ rank.}
The conclusion also holds when $\operatorname{rank}_d(Q_i)\leq r$ for
every coordinate, with $c_{d,r}=\delta_{d,r}^6$.
Affine pullback preserves the same representation by $r$ degree-$d$
polynomials, so \Cref{thm:ks-acceptance-probabilities} shows that each
coordinate mean on every cube differs from $1/3$ by at least
$\delta_{d,r}$. The preceding calculation applies.
This proves the bounded-$\operatorname{rank}_d$ assertion
in~\cite[Section~1.2]{KS26}.

\begin{corollary}[Other non-dyadic Bernoulli parameters]
\label{cor:non-dyadic-product-cube}
For every fixed $d\ge1$ and every non-dyadic $p\in(0,1)$, there is
$c_{d,p}>0$ such that every degree-at-most-$d$ map
$Q:\F_2^s\to\F_2^n$ satisfies
\[
 \Ov\left(Q(\U_s),
       \Ber(p)^{\otimes n}\right)
 \le\exp(-c_{d,p}n).
\]
For every fixed integer $r\geq1$, the same conclusion holds when
$\operatorname{rank}_d(Q_i)\leq r$ for each coordinate, with a constant
$c_{d,r,p}>0$.
\end{corollary}

\begin{proof}
Use the constant $\delta_{d,p}>0$ from~\cite[Theorem~A.8]{KS26}, restated in
\Cref{thm:ks-acceptance-probabilities}. Affine restriction preserves degree, so choose any fixed
power of two $K>p(1-p)/\delta_{d,p}^2$ and apply
\Cref{thm:affine-cube-amplification} with $\delta=\delta_{d,p}$.
For the bounded-rank extension, expanding the Boolean function
$\Gamma$ in its multilinear representation gives $\deg Q_i\leq dr$.
The first part of this corollary therefore applies with degree parameter
$dr$, giving the desired bound with $c_{d,r,p}=c_{dr,p}>0$.
\end{proof}

\begin{corollary}[Quadratic polynomial sources]
\label{cor:unrestricted-quadratic-one-third-cube}
For every $s,n\ge1$ and every quadratic map
$Q:\F_2^s\to\F_2^n$, including arbitrary affine and constant coordinates,
\begin{equation}
 \Ov\left(Q(\U_s),
       \Ber(1/3)^{\otimes n}\right)
 \le\exp(-2^{-26}n).
 \label{eq:unrestricted-quadratic-exp-one-third}
\end{equation}
Equivalently, the statistical distance is at least
$1-\exp(-2^{-26}n)$.  The same bound holds for every translate of the
product target distribution.
\end{corollary}

\begin{proof}
Every affine restriction of a quadratic Boolean function is quadratic.
The inequality $|\Pr[q(\U_t)=1]-1/3|\geq1/24$ for quadratic $q$
is also recorded in~\cite[Section~6.1]{KS26}.
For completeness, let $B$ be the polar bilinear form of a quadratic
$q:\F_2^t\to\F_2$ and let $R$ be its radical.  If
$b=\E_x(-1)^{q(x)}$, averaging the derivative gives
\[
 b^2=\E_h(-1)^{q(h)+q(0)}\mathbf1_{h\in R}.
\]
The function $q(h)+q(0)$ is linear on $R$, so the right side is either
zero or $2^{-\operatorname{rank}B}$.  The rank of an alternating form is
even.  Therefore
\[
 \Pr[q(\U_t)=1]\in\left\{\frac12\right\}
 \cup\left\{\frac{1\pm2^{-a}}2:a\in\mathbb Z_{\ge0}\right\}.
\]
The two admissible values bracketing $1/3$ are $1/4$ and $3/8$.
Thus every such probability differs from $1/3$ by at least $1/24$.
Apply \Cref{thm:affine-cube-amplification} with
$p=1/3$, $\delta=1/24$, $t=9$, $K=512$, and $m=2/3$.
Then
\[
 \delta^2-\frac{p(1-p)}K
 =\frac1{576}-\frac1{2304}=\frac1{768},
 \qquad
 \frac2{m^2K}\left(\frac1{768}\right)^2=2^{-26}.
\]
Complementing the appropriate output coordinates preserves quadratic
degree and proves the translated assertion.
\end{proof}

%% file: ks_cube_entropic_corollary.tex
\subsection{Hellinger Affinity and Relative Entropy}
\label{sec:cube-hellinger-entropy}

\begin{corollary}[Bounds on Hellinger affinity and relative entropy]
\label{cor:cube-hellinger-entropy-separation}
Let $d\geq1$ and let $Q:\F_2^s\to\F_2^n$ have degree at most $d$. Put
$P=Q(\U_s)$ and $M=\Ber(1/3)^{\otimes n}$, and let
$c_d$ be as in \Cref{thm:unrestricted-fixed-degree-one-third-cube}.
Then
\[
 \Bc(P,M):=\sum_y\sqrt{P(y)M(y)}
 \le e^{-c_dn/2}.
\]
Moreover, with natural logarithms, every input distribution $\nu$ satisfies
\begin{align*}
 s\ln2-H(\nu)+\Tc(Q(\nu))
 +\sum_{i=1}^n d(q_i^\nu\Vert1/3)
 &=D_{\Kl}(\nu\Vert\U_s)+D_{\Kl}(Q(\nu)\Vert M)\\
 &\ge c_dn,
\end{align*}
where $H(\nu)=-\sum_x\nu(x)\ln\nu(x)$,
$q_i^\nu=\Pr_\nu[Q_i=1]$, $d$ denotes binary relative entropy,
and $\Tc(R)=D_{\Kl}(R\Vert\bigotimes_iR_i)$.
For quadratic maps one may take $c_2=2^{-26}$.
\end{corollary}

\begin{proof}
For any two laws, Cauchy--Schwarz gives
$\Bc^2\le\Ov(2-\Ov)
\le2\Ov$.
Apply \Cref{thm:unrestricted-fixed-degree-one-third-cube} to $m$
independent copies of $Q$.  Their concatenation still has degree at most
$d$, while Hellinger affinity is multiplicative.  Thus
$\Bc(P,M)^{2m}\le2e^{-c_dnm}$. Taking $m$th roots and
letting $m\to\infty$ proves the affinity bound.  The quadratic constant
follows in the same way from
\Cref{cor:unrestricted-quadratic-one-third-cube}.

The variational formula for R\'enyi divergence at order $1/2$
\cite[Theorem~30]{vanErvenHarremoes14}, combined with the chain rule for
relative entropy, gives the following identity, whose short proof we include:
\[
 \min_\nu\{D_{\Kl}(\nu\Vert\U_s)
             +D_{\Kl}(Q(\nu)\Vert M)\}
 =-2\ln\Bc(P,M).
\]
Indeed, for $R=Q(\nu)$ the chain rule for relative entropy applied to $Q$ gives
$D_{\Kl}(\nu\Vert\U_s)\ge D_{\Kl}(R\Vert P)$,
with equality when $\nu$ is uniform within each fiber.  If
$G(y)=\sqrt{P(y)M(y)}/\Bc(P,M)$, then
\[
 D_{\Kl}(R\Vert P)+D_{\Kl}(R\Vert M)
 =2D_{\Kl}(R\Vert G)-2\ln\Bc(P,M).
\]
The minimum is attained by lifting $G$ uniformly over the fibers.
Finally, decompose divergence from the product law $M$ into total
correlation and marginal divergences.
\end{proof}

%% file: cube_local_hierarchy.tex
\section{An Explicit Hierarchy with Optimal Entropy Dependence}
\label{sec:cube-local-hierarchy}

\subsection{Separation from Products of AND Outputs}
\label{sec:hierarchy-and-products}

To separate degree $d+1$ from degree $d$, we use
$\Ber(2^{-(d+1)})^{\otimes n}$, which can be sampled by disjoint AND
gates on $d+1$ inputs. To obtain the quantitative bound in
Eqn.~\eqref{eq:intro-and-product}, we replace the Pinsker estimate in the
proof of \Cref{thm:affine-cube-amplification} with the following bound
that also uses the variance of the coordinate functions.

\begin{lemma}[An expectation bound from relative entropy and variance]
\label{lem:entropy-bounded-variance}
Let $M$ be a strictly positive product probability measure on a finite
product space, and let $G=\sum_i g_i$ be a sum of functions
of its separate coordinates, with $\E_Mg_i=0$ and $|g_i|\leq1$.
Put $V=\sum_i\E_Mg_i^2$. For every probability measure $R\ll M$, writing
$D=D_{\Kl}(R\|M)$ with natural logarithms, we have
\[
 \E_RG\leq2\sqrt{VD}+D.
\]
\end{lemma}

\begin{proof}
For $0<\lambda\leq1$, the inequality
$e^u\leq1+u+u^2$ for $|u|\leq1$ and independence give
$\ln\E_Me^{\lambda G}\leq V\lambda^2$.
Nonnegativity of relative entropy with respect to the probability measure
proportional to $e^{\lambda G}M$ yields
$\lambda\E_RG\leq D+V\lambda^2$.
For $D,V>0$, take
$\lambda=\sqrt D/(\sqrt D+\sqrt V)$. Then
$D/\lambda+V\lambda\leq D+2\sqrt{VD}$.
The cases $D=0$ or $V=0$ are immediate.
\end{proof}

\begin{theorem}[Separation from products of AND outputs]
\label{thm:cube-and-product-separation}
For every integer $d\geq1$, put $p=2^{-(d+1)}$.
For every $n\geq1$, $s\geq0$, and every degree-$d$ map
$Q:\F_2^s\to\F_2^n$,
\begin{equation}
 \Ov\bigl(Q(\U_s),\Ber(p)^{\otimes n}\bigr)
 \leq\exp\bigl(-2^{-3d-10}n\bigr).
 \label{eq:cube-and-product-separation}
\end{equation}
\end{theorem}

\begin{proof}
A nonzero Boolean polynomial of degree at most $d$ has mean at least
$2^{-d}$. To see this, choose a monomial of maximum degree $r\leq d$.
On every fixing of the other variables its coefficient remains one, so
the restriction to these $r$ variables is nonzero and is one at some
point. Averaging gives mean at least $2^{-r}\geq2^{-d}$.
Thus every affine restriction of every coordinate of $Q$ has mean either
zero or at least $2p$, and hence has distance at least $p$ from $p$.

Use the probability measure $\Pi$ on affine-cube parameters from the proof of
\Cref{thm:affine-cube-amplification}, now with
\[
 K=2^{d+2}=2/p,\qquad
 L=\ln\frac1{\Ov(Q(\U_s),\Ber(p)^{\otimes n})},
 \qquad v=p(1-p).
\]
Under $\Pi$, each output vector $Y_v=Q(X_v)$ has relative entropy at
most $KL$ from $M=\Ber(p)^{\otimes n}$. For distinct labels $v,w$, the
joint distribution of $(Y_v,Y_w)$ has relative entropy at most $KL$
from $M\otimes M$.
The function $(a,b)\mapsto(a-p)(b-p)$ has mean zero and variance $v^2$
under $\Ber(p)^{\otimes2}$. The function $a\mapsto(a-p)^2-v$ has mean
zero and variance $(1-2p)^2v\leq v$ under $\Ber(p)$. Both functions
have absolute value at most one.
Apply \Cref{lem:entropy-bounded-variance} across the $n$ coordinates to
the laws of $Y_v$ and $(Y_v,Y_w)$ for $v\ne w$. Expanding the squared
mean of each output coordinate over the cube, as in the proof of
\Cref{thm:affine-cube-amplification}, gives
\begin{equation}
 np^2\leq\frac{nv}{K}
       +2\left(v+\frac{\sqrt v}{K}\right)\sqrt{nKL}+KL.
 \label{eq:cube-and-product-variance}
\end{equation}
For every cube and coordinate, the squared difference between its mean
and $p$ is at least $p^2$, which gives the lower bound $np^2$.

Since $v\leq p$ and $K=2/p$, the first term is at most $np^2/2$, and
$2(v+\sqrt v/K)\leq3p$.
If $L\leq np^3/128$, then $KL\leq np^2/64$, so the right side of
Eqn.~\eqref{eq:cube-and-product-variance} is at most
\[
 np^2\left(\frac12+\frac38+\frac1{64}\right)
 =\frac{57}{64}np^2,
\]
a contradiction. Therefore $L>np^3/128=2^{-3d-10}n$, which proves
Eqn.~\eqref{eq:cube-and-product-separation}.
\end{proof}

\subsection{A Hierarchy at Prescribed Entropy}
\label{sec:hierarchy-prescribed-entropy}

\begin{corollary}[An adjacent-degree hierarchy at prescribed entropy]
\label{cor:cube-and-block-hierarchy}
Let $d,N,k$ be integers with $d\geq1$ and $0\leq k\leq N$, and put
\[
 q=\min\left\{\left\lfloor\frac{k}{d+1}\right\rfloor,N-k\right\}.
\]
If $q\geq1$, there is a uniformly and deterministically constructible
map $F_{d,N,k}:\F_2^k\to\F_2^{N-k}$ of degree exactly $d+1$ such that
its graph distribution
$D_{d,N,k}=(\U_k,F_{d,N,k}(\U_k))$ is flat of entropy exactly $k$ and
\begin{equation}
 \sup_{X\in\PolySrc_d(N)}\Ov(X,D_{d,N,k})
 \leq\exp\bigl(-2^{-3d-10}q\bigr).
 \label{eq:cube-all-entropy-bound}
\end{equation}
Every output coordinate depends on at most $d+1$ seed bits, and the
sampler uses $O(N)$ field operations, uniformly in $d$.
For $s_0=\min\{k,N-k\}\geq2(d+1)$, the right side is at most
\[
 \exp\left(-\frac{2^{-3d-11}}{d+1}s_0\right).
\]
For every flat $N$-bit distribution $D$ of entropy $k$, conversely,
\[
 \sup_{X\in\PolySrc_d(N)}\Ov(X,D)
 \geq2^{-\min\{k,N-k\}}.
\]
\end{corollary}

\begin{proof}
Divide the first $(d+1)q$ seed bits into $q$ disjoint blocks of size $d+1$.
The first $q$ output coordinates of $F_{d,N,k}$ are the products of the
bits in these blocks. All remaining coordinates are zero. The sampler
$x\mapsto(x,F_{d,N,k}(x))$ outputs all $k$ seed bits, so it is injective
and its output distribution has entropy exactly $k$.
At least one product is present, so its degree is exactly $d+1$.
The locality and deterministic construction follow directly from the
definition. Computing the $q$ products takes $dq\leq k\leq N$ field
multiplications, so sampling uses $O(N)$ field operations.

Project onto the $q$ AND outputs. The target distribution becomes
$\Ber(2^{-(d+1)})^{\otimes q}$, while the projection of
any degree-$d$ source is still a degree-$d$ source.
Postprocessing monotonicity of overlap and
\Cref{thm:cube-and-product-separation} prove
Eqn.~\eqref{eq:cube-all-entropy-bound}.
Moreover, $q\geq\lfloor s_0/(d+1)\rfloor\geq s_0/(2(d+1))$ when
$s_0\geq2(d+1)$.

For the converse, a point mass at any point in the support of $D$ has
overlap $2^{-k}$ with $D$, while $\U_N$ has overlap $2^{k-N}$.
Both are degree-$d$ sources. Taking the larger proves the lower bound.
\end{proof}

\subsection{Growing Degrees}
\label{sec:hierarchy-growing-degrees}

The same construction applies when $d$ grows logarithmically with $N$.

\begin{corollary}[Growing degrees]
\label{cor:cube-growing-degree}
Fix $0<\varepsilon<1$. For all sufficiently large $N$, uniformly for integers
\[
 1\leq d\leq\frac{1-\varepsilon}{3}\log_2N,
\]
the construction in \Cref{cor:cube-and-block-hierarchy} with
$k=\lfloor N/2\rfloor$ gives a flat degree-$(d+1)$ target distribution of entropy
$k$, locality $d+1$, and overlap
\[
 \sup_{X\in\PolySrc_d(N)}\Ov(X,D_{d,N,k})
 \leq\exp\bigl(-N^{\varepsilon-o(1)}\bigr).
\]
The construction is deterministic and runs in polynomial time uniformly
in $N$ and $d$.
\end{corollary}

\begin{proof}
Here $q=\lfloor\lfloor N/2\rfloor/(d+1)\rfloor
\geq N/(4(d+1))$ for sufficiently large $N$, uniformly in the stated
range. Consequently
\[
 2^{-3d-10}q
 \geq\frac{N^\varepsilon}{2^{12}(d+1)}
 =N^{\varepsilon-o(1)}.
\]
Apply Eqn.~\eqref{eq:cube-all-entropy-bound}. The sampler uses only disjoint
products, uniform seed coordinates, and zero padding, so its description
and evaluation remain polynomial in $N$ when $d$ grows.
\end{proof}